\documentclass[leqno,abstract=true]{scrartcl}
\usepackage{epic,latexsym,amsbsy,amsmath,amscd,amssymb}
\usepackage{authblk}
\usepackage{amsthm}
\usepackage{mathtools}
\usepackage{ bbold }
\usepackage{subcaption}
\usepackage{graphicx}
\usepackage{tikz}
\usetikzlibrary{patterns}
\usetikzlibrary{positioning,arrows}
\usetikzlibrary{decorations.markings}
\usetikzlibrary{decorations.pathmorphing}
\usetikzlibrary{arrows.meta}
\usetikzlibrary{calc}
\tikzstyle{every node}=[font=\small]
\usepackage[font=small]{caption}
\usepackage{adjustbox}
\usepackage{hyperref}

\usepackage[utf8]{inputenc}
\usepackage{enumitem}
\setenumerate{itemsep=1pt}

\begin{document}

\newtheorem{thm}{Theorem}[section]
\newtheorem{pro}[thm]{Proposition}%[section]
\newtheorem{lem}[thm]{Lemma}%[section] 
\newtheorem{cor}[thm]{Corollary}%[section]

\theoremstyle{definition}
\newtheorem{con}[thm]{Convention}%[section]
\newtheorem{dfn}[thm]{Definition}%[section]
\newtheorem{clm}[thm]{Claim}%[section]
\newtheorem{exa}[thm]{Example}%[section]
\newtheorem{exs}[thm]{Examples}%[section]
\newtheorem{rmk}[thm]{Remark}%[section]

\linespread{1.1}

%    Absolute value notation
\newcommand{\abs}[1]{\left\lvert#1\right\rvert}

\newcommand\Si{\mathbf{\Sigma}}
\newcommand\Non{\mathbf{N}}
\newcommand\Pro{\mathbf{P}}
\newcommand\Ccal{\mathcal{C}}
\newcommand\Tcal{\mathcal{T}}
\newcommand\Gcal{\mathcal{G}}
\newcommand\equi{\; \Leftrightarrow \;}
\newcommand\maxi{\mathrm{max}}

\title{Comparing consecutive letter counts in multiple context-free languages}
\author{Florian Lehner\thanks{F.\ Lehner was supported by FWF (Austrian Science Fund) project J3850-N32}\; and Christian Lindorfer\thanks{C.\ Lindorfer was partially supported by FWF (Austrian Science Fund) projects P31237 and DK W1230.}}
\affil{\normalsize Institut f\"ur Diskrete Mathematik, 
Technische Universit\"at Graz,
%Steyrergasse 30, A-8010 Graz, 
Austria}

%\address{\parbox{.8\linewidth}{Institut f\"ur Diskrete Mathematik,\\ 
%Technische Universit\"at Graz,\\
%Steyrergasse 30, A-8010 Graz, Austria\\ $\,$}}

%\email{mail@florian-lehner.net, lindorfer@math.tugraz.at}

\date{\today} 
%\thanks{Partially supported by 
%Austrian Science Fund project FWF-P31889-N35}
%\subjclass[2010] {20F10, % Word problems, 
          %other decision problems, connections with logic and automata
%                  68Q45,  %Formal languages and automata
%                  05C25}   %Graphs and groups.
		  
%\keywords{Self-avoiding walks,  Cayley graph, graph decompositions}

%\maketitle

% \markboth{{\sf F. Lehner and C. Lindorfer}}
% {{\sf }}
% \baselineskip 15pt

%%%%%%%%%%%%%%%%%%%%%%%%%%%%%%%%%%%

\maketitle

\begin{abstract}
Context-free grammars are not able to capture cross-serial dependencies occurring in some natural languages. To overcome this issue, Seki et al.\ introduced a generalization called $m$-multiple context-free grammars ($m$-MCFGs), which deal with $m$-tuples of strings. We show that $m$-MCFGs are capable of comparing the number of consecutive occurrences of at most $2m$ different letters. In particular, the language $\{a_1^{n_1} a_2^{n_2} \cdots a_{2m+1}^{n_{2m+1}} \mid n_1 \geq n_2 \geq \dots \geq n_{2m+1} \geq 0\}$ is $(m+1)$-multiple context-free, but not $m$-multiple context-free.
\end{abstract}

\section{Introduction}

Formal language theory makes use of mathematical tools to study the syntactical aspects of natural and artificial languages. Two of the best known and most studied classes of formal languages are context free languages and context sensitive languages, generated by context free grammars and context sensitive grammars, respectively. Context-free grammars have convenient generative properties, but they are not able to model cross-serial dependencies, which occur in Swiss German and a few other natural languages. The expressive power of context-sensitive grammars on the other hand often exceeds our requirements, and the decision problem whether a given string belongs to the language generated by such a grammar is PSPACE-complete.

To overcome these issues, intermediate classes of `mildly context sensitive languages' were independently introduced by {\scshape Vijay-Shanker et al.~\cite{VS87}} and {\scshape Seki et al.~\cite{Se91}} in the form of context-free rewriting systems and multiple context-free grammars (MCFGs). These concepts  turn out to be equivalent in the sense that they both lead to the same class of languages, called multiple context-free languages (MCFLs). While MCFGs are able to model cross-serial dependencies by dealing with tuples of strings, the languages generated by them share several important properties with context free languages, such as polynomial time parsability and semi-linearity. 

MCFLs can be distinguished depending on the largest dimension $m$ of tuples involved. The $m$-MCFLs obtained in this way form an infinite strictly increasing hierarchy 
\[\text{CFL} = \text{1-MCFL} \subsetneq \text{2-MCFL} \subsetneq \ldots \subsetneq \text{$m$-MCFL} \subsetneq \text{$(m+1)$-MCFL} \subsetneq \ldots \subsetneq \text{CSL},\]

where CFL and CSL denote the classes of context free languages and context sensitive languages, respectively.

A highlight in the theory of MCFGs is a result by {\scshape Salvati~\cite{Sa15}}, stating that the language $O_2=\{w \in \{a,\bar{a},b,\bar{b}\}^* \mid \abs{w}_a=\abs{w}_{\bar{a}} \wedge \abs{w}_b=\abs{w}_{\bar{b}}\}$ occurring as the word problem of the group $\mathbb{Z}^2$ is a 2-MCFL. Moreover the language $\mathrm{MIX}=\{w \in \{a,b,c\}^* \mid \abs{w}_a=\abs{w}_b=\abs{w}_c\}$ is rationally equivalent to $O_2$ and thus also a 2-MCFL. {\scshape Ho~\cite{Ho18}} generalized this result by showing that for any positive integer $d$ the word problem of $\mathbb{Z}^d$ is multiple context-free.

In this paper we study languages defined by comparing lengths of runs of consecutive identical letters and show that they are able to separate the layers of the hierarchy mentioned above. In particular we consider languages of the form
\[
L_k=\{a_1^{n_1} a_2^{n_2} \cdots a_{k}^{n_k} \mid n_1 \geq n_2 \geq \dots \geq n_k \geq 0\}
\]
and generalisations thereof. The languages $L_1$ and $L_2$ are easily seen to be context-free, and it is a standard exercise to show that $L_3$ is not context-free by using the pumping lemma for context free languages. Our main result generalises these observations.

\begin{thm}
\label{thm:main}
The language $L_k=\{a_1^{n_1} a_2^{n_2} \cdots a_{k}^{n_k} \mid n_1 \geq n_2 \geq \dots \geq n_k \geq 0\}$ is a $\lceil k/2\rceil$-MCFL but not a $(\lceil k/2\rceil-1)$-MCFL.
\end{thm}

The first part of Theorem~\ref{thm:main} is verified by constructing an appropriate grammar. For the second part, one might hope that it is implied by a suitable generalisation of the pumping lemma to $m$-MCFLs, but unfortunately such a generalisation does not exist.

A weak pumping lemma for $m$-MCFLs due to {\scshape Seki et al.~\cite{Se91}} which generalises pumpability of words to $m$-pumpability only confirms the existence of $m$-pumpable strings in infinite $m$-MCFLs and not that all but finitely many words in the language are $m$-pumpable. In particular, it is not strong enough to imply the second part of Theorem~\ref{thm:main}. While {\scshape Kanazawa~\cite{Ka09}} managed to prove a strong version of the pumping lemma for the sub-class of well-nested $m$-MCFLs, {\scshape Kanazawa et al.~\cite{Ka14}} showed that in fact such a pumping lemma cannot exist for general $m$-MCFLs by giving a 3-MCFL containing infinitely many words which are not $k$-pumpable for any given $k$. Nevertheless, our proof relies heavily on the idea of pumping, thus showing that this technique can be useful even in cases where no pumping lemma is available.

\section{Definitions and notation}

For an \emph{alphabet} (finite set of letters) $\Si$ we denote by 
\[
\Si^*=\{w=a_1 a_2 \cdots a_n \mid n \geq 0, a_i \in \Si\}
\]
the set of all finite words over $\Si$. A \emph{formal language} over $\Si$ is a subset of $\Si^*$.

The \emph{length} $|w|$ of a word $w = a_1a_2\cdots a_n$ is the number $n$ of letters contained in it. We write $\epsilon$ for the word of length zero and $a^n$ for the word obtained by $n$-fold repetition of the letter $a$.  
 
In this paper we focus on languages based on comparing consecutive occurrences of different letters. In order to formally define these languages, we need some preliminary definitions.
A \emph{preorder} $\preceq$ on a set $M$ is a reflexive and transitive binary relation on $M$. In contrast to partial orders, preorders need not be antisymmetric, that is, $a \preceq b$ and $b\preceq a$ may be true at the at the same time for different elements $a,b$. A preorder $\preceq$ is called \emph{total} if for all $a,b \in M$ at least one of $a \preceq b$ and $b \preceq a$ holds. The \emph{comparability graph} of a preorder is the simple undirected graph with vertex set $M$, where two different vertices $u$ and $v$ are connected by an edge if they are comparable. We call a preorder \emph{connected}, if its comparability graph is connected. Note that any total preorder is connected, but a connected preorder does not have to be total.

For a positive integer $m$ and a preorder $\preceq$ on $[m]:=\{1,2,\dots,m\}$ define the language $L_\preceq$ over the alphabet $\Sigma=\{a_1, \dots, a_m\}$ by
\[L_\preceq=\{a_1^{n_1} a_2^{n_2} \cdots a_m^{n_m} \mid i \preceq j \Rightarrow n_i \leq n_j\}.\]

A preorder $\preceq'$ on $M$ is said to be a \emph{totalisation} of a preorder $\preceq$ on $M$, if it is total and extends $\preceq$, that is, whenever $a \preceq b$ also $a \preceq' b$. Let $T_{\preceq}$ be the set of totalisations of $\preceq$.

\begin{rmk} \label{rmk:precequnion}
Let $w=a_1^{n_1} a_2^{n_2} \cdots a_m^{n_m} \in L_\preceq$. The binary relation $\preceq'$ on $[m]$ defined by $i \preceq' j$ if $n_i \leq n_j$ is a totalisation of $\preceq$. Consequently,
%Observe that
\[L_\preceq=\bigcup_{\preceq' \in T_\preceq} L_{\preceq'}.\]
%This is a consequence of the fact that for any given word $w=a_1^{n_1} a_2^{n_2} \cdots a_m^{n_m} \in L_\preceq$, the binary relation $\preceq'$ on $[m]$ defined by $i \preceq' j$ if and only if $n_i \leq n_j$ is a totalisation of $\preceq$. %Der Satz ist nicht gut!
\end{rmk}

A natural way of specifying a language is by giving a grammar which generates it. In this paper we focus on multiple context-free languages and their generating grammars, which we shall now define.

Let $\Si$ be an alphabet and $\Non$ be a finite ranked set of non-terminals, that is, $\Non$ is the disjoint union $\Non=\bigcup_{r \in \mathbb{N}} \Non^{(r)}$ of finite sets $\Non^{(r)}$. Note that since $\Non$ is finite, all but finitely many $N^{(r)}$ must be empty. The elements of $N^{(r)}$ are called \emph{non-terminals of rank $r$}. A \emph{production rule} $\rho$ over $(\Non, \Si)$ is an expression 
\[
A(\alpha_1, \dots, \alpha_r) \leftarrow A_1(x_{1,1}, \dots, x_{1,r_1}), \dots, A_n(x_{n,1}, \dots, x_{n,r_n}),
\]
where
\begin{enumerate}[label=(\roman*)]
    \item $n \geq 0$,
    \item $A \in \Non^{(r)}$ and $A_i \in N^{(r_i)}$ for all $i \in [n]$,
    \item $x_{i,j}$ are variables,
    \item $\alpha_1, \dots, \alpha_r$ are strings over $\Si \cup \{x_{i,j} \mid i \in [n], j \in [r_i]\}$, such that each $x_{i,j}$ occurs at most once in $\alpha_1 \cdots \alpha_r$.
\end{enumerate}
Production rules satisfying $n=0$ are called \emph{terminating rules}.

For $A \in \Non^{(r)}$ and words $w_1, \dots, w_r \in \Si^*$ we call $A(w_1, \dots, w_r)$ a \emph{term}. Let $\rho$ be a production rule as above. The \emph{application} of $\rho$ to a sequence of $n$ terms $(A_i(w_{i,1}, \dots, w_{i,r_i}))_{i \in [n]}$ yields the term $A(w_1, \dots, w_r)$, where $w_l$ is obtained from $\alpha_l$ by substituting every variable $x_{i,j}$ by the word $w_{i,j}$ for $l \in [r]$.

A \emph{multiple context-free grammar} is a quadruple $\mathcal{G}=(\Non,\Si,\Pro,S)$, where $\Non$ is a finite ranked set of non-terminals, $\Si$ is an alphabet, $\Pro$ is a finite set of production rules over $(\Non,\Si)$ and $S \in \Non^{(1)}$ is the start symbol. The grammar $\Gcal$ is \emph{$m$-multiple context-free} or a $m$-MCFG, if the rank of all non-terminals in $\Non$ is at most $m$.

We call a term $T$ \emph{derivable} in $\Gcal$ and write $\vdash T$ if there is a rule $\rho$ and a sequence of derivable terms $\mathcal{A}$ such that the application of $\rho$ to $\mathcal{A}$ yields $T$. Note that if $\rho=A(w_1, \dots, w_r) \leftarrow$ is a terminating rule, then $\mathcal{A}$ is the empty sequence. Thus the term $A(w_1, \dots, w_r)$ is derivable.

The language generated by $\Gcal$ is the set $L(\Gcal)=\{w \in \Si^* \mid \; \vdash S(w)\}$. We call a language $m$-\emph{multiple context-free} or an $m$-MCFL, if it is generated by an $m$-MCFG.

By the following lemma it is enough to consider MCFGs in a certain normal form.

\begin{lem}[{{\scshape Seki et al.~\cite[Lem.~2.2]{Se91}}}] \label{lem:normalform}
Every $m$-MCFL is generated by an $m$-MCFG satisfying the following conditions.
\begin{enumerate}[label=(\roman*)]
    \item If $A(\alpha_1, \dots, \alpha_r) \leftarrow A_1(x_{1,1}, \dots, x_{1,r_1}), \dots, A_n(x_{n,1}, \dots, x_{n,r_n})$ is a non-terminating rule, then the string $\alpha_1 \cdots \alpha_r$ contains each $x_{i,j}$ exactly once and does not contain elements of $\Si$.
    \item If $A(w_1, \dots, w_r) \leftarrow$ is a terminating rule, then the string $w_1 \cdots w_r$ contains exactly one letter of $\Si$.
\end{enumerate}
\end{lem}

A \emph{rooted tree} $T$ is a tree with a designated root vertex. The \emph{descendants} of a vertex $v$ of $T$ are all vertices $u$ such that if $v$ lies on the unique shortest path from $u$ to the root of $T$. Descendants of $v$ adjacent to $v$ are called \emph{children} of $v$.  A rooted tree is called ordered, if an ordering is specified for the children of each vertex. The \emph{subtree rooted at} a vertex $v$ of $T$ is the subgraph of $T$ consisting of $v$ and its descendants and all edges incident to these descendants.

Derivation trees for multiple context-free languages were first defined by {\scshape Seki et al.~\cite{Se91}}; we will use a slight variation of their definition. Let $\Gcal=(\Non,\Si,\Pro,S)$ be an MCFG. An ordered rooted tree $D$ whose vertices are labelled with elements of $\Pro$ is a \emph{derivation tree} of a term $T$, if the tree and its labelling satisfy the following conditions.
%if it satisfies one of the following conditions.

\begin{enumerate}[label=(\roman*)]
    \item The root of $D$ has $n \geq 0$ children and is labelled with a rule $\rho \in \Pro$. 
    \item For $i \in [n]$ the subtree $D_i$ rooted at the $i$-th child of the root of $D$ is a derivation tree of a term $T_i$.
    %The root of $D$ has no children and is labelled by a terminating rule $A(w_1, \dots, w_r) \leftarrow \; \in \Pro$.
    \item The rule $\rho$ applied to the sequence $(T_i)_{i \in [n]}$ yields $T$.  
\end{enumerate}

It is not hard to see that $\vdash A(w_1, \dots, w_r)$ if and only if there is a derivation tree $D$ of $A(w_1, \dots, w_r)$. However, in general such a derivation tree need not be unique. We denote by $\ell(D)$ the label of the root of $D$. %We denote by $w(D)$ the tuple of strings generated by $D$, that is the tuple such that $D$ is a derivation tree of $w(D)$.
\begin{rmk}\label{rmk:subs}
Let $D$ be a derivation tree and let $v$ be a vertex of $D$. Then by definition replacing the subtree $D'$ of $D$ rooted at $v$ by a derivation tree $D''$ satisfying $\ell(D'')=\ell(D')$ yields a derivation tree.
\end{rmk}

\section{Main result}

We split the proof of our main result into two parts, covered by Theorem~\ref{thm:main_general1} and Theorem~\ref{thm:main_general2}, respectively. Together, these two results clearly imply Theorem~\ref{thm:main}; it is also worth pointing out that in fact they cover the (much larger) class of languages $L_\preceq$ as introduced in the previous section.

\begin{thm}
\label{thm:main_general1}
For every preorder $\preceq$ the language $L_\preceq=\{a_1^{n_1} a_2^{n_2} \cdots a_m^{n_m} \mid  i \preceq j \Rightarrow n_i \leq n_j \}$ over the alphabet $\Si=\{a_1, \dots, a_m\}$ is a $\lceil m/2\rceil$-MCFL.
\end{thm}
\begin{proof}
It is known (see for instance \cite{Se91}) that the class of $k$-MCFLs is a full AFL; in particular it is closed under substitution and taking finite unions. Thus it is enough to consider the case where $m=2k$ is even, the case $m=2k-1$ follows by substituting $\epsilon$ for $a_{2k}$. Additionally, by Remark~\ref{rmk:precequnion} we may assume that $\preceq$ is a total preorder. 

We show that $L_\preceq$ is generated by the $k$-MCFG $\Gcal=(\Non=\{S,A\},\Si,\Pro,S)$, where $A$ has rank $k$ and $\Pro$ consists of the rules
\begin{align*}
	S(x_1 x_2 \cdots x_k) &\leftarrow A(x_1,x_2,\dots, x_k)  \\[4pt]
	A(\epsilon, \epsilon, \dots, \epsilon) &\leftarrow
\end{align*}
and for every $j \in [2k]$ the additional rule $\rho_j$ given by
\begin{equation*}\label{eq:rule}
    %A(a_1^{[j \preceq 1]} x_1 a_2^{[j \preceq 2]}, a_3^{[j \preceq 3]} x_2 a_4^{[j \preceq 4]}, \dots, a_{2n-1}^{[j \preceq 2n-1]} x_n a_{2n}^{[j \preceq 2n]}) \leftarrow A(x_1, x_2, \dots, x_n),
	%A(a_1^{\mathbb{1}_j(1)} x_1 a_2^{\mathbb{1}_j(2)}, a_3^{\mathbb{1}_j(3)} x_2 a_4^{\mathbb{1}_j(4)}, \dots, a_{2k-1}^{\mathbb{1}_j(2k-1)} x_k a_{2k}^{\mathbb{1}_j(2k)}) \leftarrow A(x_1, x_2, \dots, x_k),
	%A(\alpha_{j1}, \alpha_{j2}, \dots, \alpha_{jn}) \leftarrow A(x_1, x_2, \dots, x_n),
	A(y_{1} x_1 y_{2}, y_{3} x_2 y_{4}, \dots, y_{2n-1} x_n y_{2n})\leftarrow A(x_1, x_2, \dots, x_n),
\end{equation*}
where %$\mathbb{1}_j$ is the indicator function of the set $\uparrow \!\! j$.
%where $\alpha_{ji}=a_{2i-1} x_i a_{2i}$
%y_{j,1} x_1 y_{j,2}, y_{j,3} x_2 y_{j,4}, \dots, y_{j,2n-1} x_n y_{j,2n}
% \[
%     \alpha_{ji}=
%     \begin{cases}
%     a_{2i-1} x_i a_{2i} &\text{if $j \preceq 2i-1$ and $j\preceq 2i$,}\\
%     a_{2i-1} x_i &\text{if $j \preceq 2i-1$ and $j\npreceq 2i$,}\\
%     \phantom{a_{2i-1}}x_i a_{2i} &\text{if $j \npreceq 2i-1$ and $j\preceq 2i$,}\\
%     \phantom{a_{2i-1}}x_i &\text{if $j \npreceq 2i-1$ and $j\npreceq 2i$.}
%     \end{cases}
% \]
\[
y_i = 
\begin{cases}
a_i & \text{if }j\preceq i,\\
\epsilon &\text{otherwise.}
\end{cases}
\]

First note that if $\vdash A(w_1, \dots, w_k)$, then each $w_l$ has the form $w_l=a_{2l-1}^{n_{2l-1}} a_{2l}^{n_{2l}}$, and it holds that $n_i \leq n_j$ whenever $i \preceq j$. This is clearly true for $A(\epsilon,\epsilon, \dots, \epsilon)$ and it is preserved when applying the rule $\rho_j$, which adds one instance of the letter $a_j$ and every letter $a_i$ with $j \preceq i$. Every word $w$ generated by $\Gcal$ is in $L_\preceq$ since it is the concatenation $w_1 \cdots w_k$ of strings $w_l$ such that $\vdash A(w_1, \dots, w_k)$.

Next we show that any given word in $L_\preceq$ is generated by $\Gcal$. 
Assume for a contradiction that there is a word in $L_\preceq$ which is not generated by $\Gcal$. Pick such a word $w=a_1^{n_1} a_2^{n_2} \cdots a_{2k}^{n_{2k}}$ for which $n_{\maxi}=\max \{n_l \mid l\in[2k]\}$ is minimal.
As $\Gcal$ generates the empty word, $w \neq \epsilon$ and $n_{\maxi} \geq 1$. For $l \in [2k]$ let $n_l'=n_l$ if $n_l < n_\maxi$, and let $n_l'=n_{\maxi}-1$ otherwise. 
Since $w \in L_\preceq$ we have $n_i' \leq n_j'$ whenever $i \preceq j$, and thus $w'=a_1^{n_1'} a_2^{n_2'} \cdots a_{2k}^{n_{2k}'} \in L_\preceq$. By minimality of $w$ the word $w'$ is generated by $\Gcal$, and in particular $\vdash A(a_1^{n_1'} a_2^{n_2'}, \dots, a_{2k-1}^{n_{2k-1}'} a_{2k}^{n_{2k}'})$. Pick some minimal $j$ with respect to $\preceq$ from the set $\{l \in [2k] \mid n_l=n_\maxi\}$. Applying the rule $\rho_j$ to $A(a_1^{n_1'} a_2^{n_2'}, \dots, a_{2k-1}^{n_{2k-1}'} a_{2k}^{n_{2k}'})$ yields $\vdash A(a_1^{{n_1}} a_2^{n_2}, \dots, a_{2k-1}^{n_{2k-1}} a_{2k}^{n_{2k}})$; consequently $\Gcal$ generates $w$, contradicting our assumption.
\end{proof}

\begin{thm}
\label{thm:main_general2}
For every connected preorder $\preceq$ the language $L_\preceq=\{a_1^{n_1} a_2^{n_2} \cdots a_m^{n_m} \mid i \preceq j \Rightarrow n_i \leq n_j \}$ over the alphabet $\Si=\{a_1, \dots, a_m\}$ is not a $(\lceil m/2\rceil-1)$-MCFL.
\end{thm}
\begin{proof}
Assume that there is a MCFG $\Gcal=(\Non,\Si,\Pro,S)$ generating $L_\preceq$, and assume that $\Gcal$ is given in normal form as in Lemma~\ref{lem:normalform}.

For a derivation tree $D$ and $i \in [m]$ denote by $\abs{D}_i$ the total number of letters $a_i$ occurring in all substrings contained in the term $\ell(D)$ and by $\abs{D}=\sum_{i=1}^m \abs{D}_i$ the combined length of all substrings. Since $\Gcal$ is in normal form, if $\ell(D)$ is not a terminating rule and $D_1, \dots, D_k$ are the derivation trees rooted at the $k$ children of the root of $D$ we have 
\begin{equation} \label{eq:normalnterm}
\abs{D}_i=\sum_{j=1}^k \abs{D_j}_i.
\end{equation}
Moreover, if $\ell(D)$ is a terminating rule, then 
\begin{equation}\label{eq:normalterm}
\abs{D}=1.
\end{equation}
Call a rule a \emph{combiner}, if its right hand side contains at least 2 non-terminals and therefore a vertex of any derivation tree labelled by $\rho$ has at least $2$ children. Note that there is an upper bound $K$ such that the right hand side of any combiner contains at most $K$ non-terminals. 
%From the definition of the normal form we get the following for any derivation tree $D$. If $\ell(D)$ is a combiner then for every $i \in [m]$ there is a child of the root such that the subtree $D'$ such that $K \abs{D'}_i \geq \abs{D}_i$. If $\ell(D)$ is neither a combiner nor a terminating rule, then the derivation tree $D'$ rooted at the unique child of the root of $D$ satisfies $\abs{D'}_i=\abs{D}_i$ for every $i \in [m]$. 

Fix $n > K^{2 C}$, where $C$ is the number of combiners in $\Pro$, and let $D$ be a derivation tree of $S(a_1^n a_2^n \cdots a_m^n)$. Then $D$ contains a path starting at the root containing at least $2C+1$ vertices labelled with combiners. If not, then \eqref{eq:normalnterm} and \eqref{eq:normalterm} imply $\abs{D}\leq K^{2C}$, contradicting our choice of $n$. In particular the path contains at least 3 vertices labelled with the same combiner $\rho$. Denote the subtrees rooted at these three vertices by $D_1$, $D_2$, and $D_3$ such that $D_3 \subseteq D_2 \subseteq D_1$.

We claim that for any $i \preceq j$ we have $\abs{D_1}_j-\abs{D_2}_j = \abs{D_1}_i-\abs{D_2}_i$, and that an analogous statement holds for $D_2$ and $D_3$.

Assume that $\abs{D_1}_j-\abs{D_2}_j > \abs{D_1}_i-\abs{D_2}_i$. By \eqref{eq:normalnterm} the derivation tree $D'$ obtained by replacing $D_1$ by $D_2$ (compare Remark~\ref{rmk:subs}) satisfies
\[\abs{D'}_j-\abs{D'}_i=\abs{D}_j-(\abs{D_1}_j-\abs{D_2}_j) - \abs{D}_i + (\abs{D_1}_i-\abs{D_2}_i)<0,\]
because $\abs{D}_j=\abs{D}_i=n$. This is a contradiction, as the word $w(D')$ is not in $L_\preceq$.
If $\abs{D_1}_j-\abs{D_2}_j < \abs{D_1}_i-\abs{D_2}_i$, then the derivation tree $D''$ obtained by replacing $D_2$ by $D_1$ satisfies
\[\abs{D''}_j-\abs{D''}_i=\abs{D}_j+(\abs{D_1}_j-\abs{D_2}_j) - \abs{D}_i - (\abs{D_1}_i-\abs{D_2}_i)<0,\]
which is a contradiction for the same reason as before. This completes the proof of our claim.

\begin{figure}
\centering
\begin{subfigure}[t]{0.3\textwidth}
    \vskip 0pt
    \centering
    \begin{tikzpicture}
    \begin{scope}[scale=0.65]
  	    \draw (0,0) -- (-3.25,-5);
  	    \draw (0,0) -- (3.25,-5);
  	    \draw (-0.5,-2) -- (-2.45,-5);
  	    \draw (-0.5,-2) -- (1.45,-5);
  	    \draw (0,-4) -- (-0.65,-5);
  	    \draw (0,-4) -- (0.65,-5);
  	    \draw (-3.25,-5) -- (3.25,-5);
  	    \fill[pattern=north west lines] (0,-4) -- (-0.65,-5) -- (0.65,-5) -- cycle;
  	    \fill[pattern=dots] (-0.5,-2) -- (-2.45,-5) -- (-0.65,-5) -- (0,-4) -- (0.65,-5) -- (1.45,-5) -- cycle;
  	    \draw (0,0) -- (0,-0.3);
  	    \draw[decorate, decoration={zigzag, segment length=6mm, amplitude=1mm}] (0,-0.3) -- (-0.5,-2);
  	    \draw (-0.5,-2) -- (-0.5,-2.3);
  	    \draw[decorate, decoration={zigzag, segment length=6mm, amplitude=1mm}] (-0.5,-2.3) -- (0,-4);
  	    %\node () at (0.4,0) {$S$};
  	    \node () at (-0.1,-2) {$\rho$};
  	    \node[fill=white, inner sep=1pt] () at (-0.4,-4) {$\rho$};
  	    \node () at (2.2,-2.3) {$D$};
  	    %\node () at (-1,-3.7) {$D_1$};
  	    %\node () at (0,-4.7) {$D_2$};
  	    \draw[white] (0,-7);
  	\end{scope}
    \end{tikzpicture}
    %\subcaption{Derivation tree $D$}
\end{subfigure}
\hfill
\begin{subfigure}[t]{0.3\textwidth}
    \vskip 0pt
    \centering
    \begin{tikzpicture}
    \begin{scope}[scale=0.65]
  	    \draw (0,0) -- (-3.25,-5);
  	    \draw (0,0) -- (3.25,-5);
  	    \draw (-0.5,-2) -- (-2.45,-5);
  	    \draw (-0.5,-2) -- (1.45,-5);
  	    \draw (-1.15,-3) -- (0.15,-3);
  	    \fill[pattern=north west lines] (-0.5,-2) -- (-1.15,-3) -- (0.15,-3) -- cycle;
  	    \draw (0,0) -- (0,-0.4);
  	    \draw (-3.25,-5) -- (-2.45,-5);
  	    \draw (1.45,-5) -- (3.25,-5);
  	    \draw[decorate, decoration={zigzag, segment length=6mm, amplitude=1mm}] (0,-0.4) -- (-0.5,-2);
  	    %\node () at (0.4,0) {$S$};
  	    \node () at (-0.1,-2) {$\rho$};;
  	    \node () at (2.2,-2.3) {$D'$};
  	    %\node () at (1.3,-2.7) {$D$};
  	    %\node () at (-0.5,-2.7) {$D_2$};
  	    \draw[white] (0,-7);
  	\end{scope}
  	\end{tikzpicture}
  	%\subcaption{Replacing $D_1$ with $D_2$ yields $D'$}
\end{subfigure}
\hfill
\begin{subfigure}[t]{0.3\textwidth}
    \vskip 0pt
    \centering
    \begin{tikzpicture}
    \begin{scope}[scale=0.65]
  	    \draw (0,0) -- (-3.25,-5);
  	    \draw (0,0) -- (3.25,-5);
  	    \draw (-0.5,-2) -- (-2.45,-5);
  	    \draw (-0.5,-2) -- (1.45,-5);
  	    \draw (0,-4) -- (-1.95,-7);
  	    \draw (0,-4) -- (1.95,-7);
  	    \draw (0.5,-6) -- (-0.15,-7);
  	    \draw (0.5,-6) -- (1.15,-7);
  	    \draw (-3.25,-5) -- (-0.65,-5);
  	    \draw (3.25,-5) -- (0.65,-5);
  	    \draw (-1.95,-7) -- (1.95,-7);
  	    \fill[pattern=north west lines] (0.5,-6) -- (1.15,-7) -- (-0.15,-7) -- cycle;
  	    \fill[pattern=dots] (-0.5,-2) -- (-2.45,-5) -- (-0.65,-5) -- (0,-4) -- (0.65,-5) -- (1.45,-5) -- cycle;
  	    \fill[pattern=dots] (0,-4) -- (-1.95,-7) -- (-0.15,-7) -- (0.5,-6) -- (1.15,-7) -- (1.95,-7) -- cycle;
  	    \draw (0,0)--(0,-0.4);
  	    \draw[decorate, decoration={zigzag, segment length=6mm, amplitude=1mm}] (0,-0.4) -- (-0.5,-2);
  	    \draw (-0.5,-2)--(-0.5,-2.4);
  	    \draw[decorate, decoration={zigzag, segment length=6mm, amplitude=1mm}] (-0.5,-2.4) -- (0,-4);
  	    \draw (0,-4)--(0,-4.4);
  	    \draw[decorate, decoration={zigzag, segment length=6mm, amplitude=1mm}] (0,-4.4) -- (0.5,-6);
  	    %\node () at (0.4,0) {$S$};
  	    \node () at (-0.1,-2) {$\rho$};
  	    \node[fill=white, inner sep=1pt] () at (-0.4,-4) {$\rho$};
  	    \node[fill=white, inner sep=1pt] () at (0.1,-6) {$\rho$};
  	    \node () at (2.2,-2.3) {$D''$};
  	    %\node () at (1.3,-2.7) {$D$};
  	    %\node () at (-1,-3.7) {$D_1$};
  	    %\node () at (0.5,-6.7) {$D_2$};
  	\end{scope}
  	\end{tikzpicture}
  	%\subcaption{Replacing $D_2$ with $D_1$ yields $D''$}
\end{subfigure}
  	\caption{Replacing $D_1$ with $D_2$ yields $D'$ and replacing $D_2$ with $D_1$ yields $D''$.}
  	\label{fig:derivtree}
\end{figure}
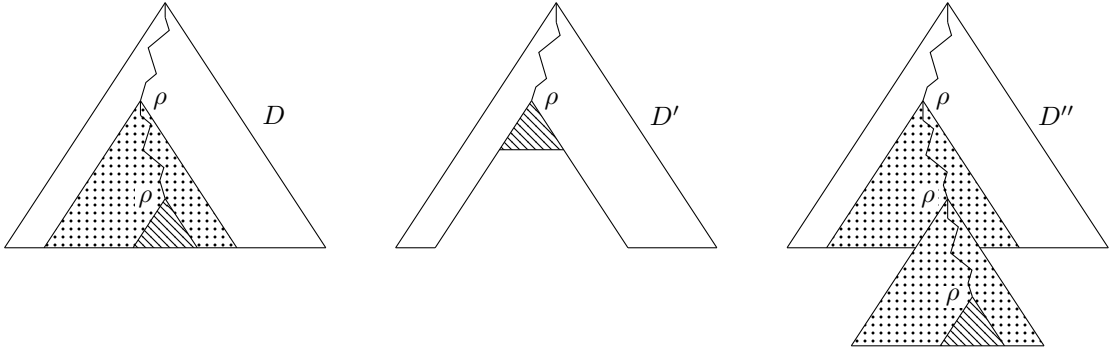

If $i,j \in [m]$ are comparable in $\preceq$, then $\abs{D_1}_j-\abs{D_1}_i = \abs{D_2}_j-\abs{D_2}_i$. By connectedness of the comparability graph this is true for any pair $i,j$. 

Since $\rho$ is a combiner, $\abs{w(D_1)}> \abs{w(D_2)}$. In particular $\abs{D_1}_i>\abs{D_2}_i$ for some and thus for every $i \in [m]$. Analogously we obtain $\abs{D_2}_i > \abs{D_3}_i$; in particular $\abs{D_2}_i >0$ holds for every $i \in [m]$. 

Assume now for a contradiction the Grammar $\Gcal$ is $(\lceil m/2 \rceil-1)$-MCF. Then $w(D_2)$ consists of at most $\lceil m/2 \rceil-1$ strings and each of them is a substring of $a_1^n a_2^n \cdots a_m^n$ because $\Gcal$ is in normal form. Every letter of $\Si$ appears in $w(D_2)$, hence one of the strings must contain at least 3 different letters and thus be of the form $a_{i-1}^{n_1} a_i^n a_{i+1}^{n_2}$ for some $i \in \{2, \dots, m-1\}$. As this contradicts the fact that $n \geq \abs{D_1}_i>\abs{D_2}_i=n$, the grammar $\mathcal{G}$ must be at least $\lceil m/2 \rceil$-MCF.
\end{proof}

\bibliographystyle{plain}
\bibliography{latex}

\end{document}